%% file: article.tex
\documentclass[submission,copyright,creativecommons]{eptcs}
\providecommand{\event}{QPL 2022} 
\usepackage{breakurl}             
\usepackage{underscore}           
\usepackage{mathtools}
\usepackage{amsmath}
\usepackage{amsthm}
\usepackage{xcolor}
\usepackage{amsfonts}
\usepackage{stmaryrd}
\usepackage{todonotes}
\definecolor{zx_grey}{RGB}{211,211,211}

\newcommand{\interp}[1]{\left\llbracket#1\right\rrbracket}

\newcommand{\N}{\mathbb{N}}

\newcommand{\df}{\stackrel{\scriptscriptstyle def}{=}}
\usetikzlibrary{circuits.ee.IEC}

\title{Large-scale quantum diagrammatic reasoning tools\\ \vspace{0.25cm} \Large !-boxes vs. scalable notations}
\author{Titouan Carette \& Louis Lemonnier
	\institute{Université Paris-Saclay, Inria, ENS Paris-Saclay, CNRS, Laboratoire Méthodes Formelles, 91190, Gif-sur-Yvette, France}
	\email{name.surname@universite-paris-saclay.fr}
}

\def\titlerunning{Large-scale quantum diagrammatic reasoning tools}
\def\authorrunning{T. Carette and L. Lemonnier}

\input{sample.tikzdefs}
\usepackage{tikzit}
\input{sample.tikzstyles}

\newcommand{\ket}[1]{|#1\rangle}
\newcommand{\bra}[1]{\langle #1|}
\newcommand{\ketbra}[1]{\ket{#1}\bra{#1}}

\newcommand{\sem}[1]{\left\llbracket #1\right\rrbracket}

\newtheorem{lemma}{Lemma}
\newtheorem{theorem}{Theorem}

\begin{document}
	\maketitle
	
	\begin{abstract}
		The application of diagrammatic reasoning techniques to large-scale quantum processes needs specific tools to describe families of diagrams of arbitrary size. For now, large-scale diagrammatic reasoning tools in ZH-calculus come in two flavours, !-boxes and scalable notations. This paper investigates the interactions between the two approaches by exhibiting correspondences through various examples from the literature, focusing on (hyper)graph states and diagrammatic transforms. In doing so, we set up a path toward a neat and tidy large-scale diagrammatic reasoning toolbox.
	\end{abstract}
	
	\section*{Introduction}
	
	\textbf{Motivation:} One of the most used models of quantum computation is quantum circuits; a circuit is a sequence of quantum logic gates applied
	to a list of qubits. It is a useful framework for writing
	quantum algorithms. However, it is challenging to rewrite circuits or
	check if two of them perform the same operation. Among the methods to reason on circuits, ZX-calculus has risen from categorical theories \cite{coecke2011interacting} to provide
	an intuitive and efficient framework for quantum computation. ZX-calculus
	contains an equational theory which allows to rewrite diagrams
	and enables smart circuit optimisations \cite{kissinger2019reducing,pyzx,backens2021extraction} and verification
	\cite{lemonnier2021hypergraph} schemes. A set of rules for ZX-calculus have been
	proven complete for Clifford+T quantum operations \cite{jeandel2018complete}.
	Its cousin ZH-calculus specialises in nonlinearity \cite{backens2018zh}, but contains some useful
	aspects \cite{kuijpers2019graphical,lemonnier2021hypergraph} 
	that are hard to transpose in ZX-calculus; as it
	carries a native representation of Toffoli and CCZ gates,
	it holds a multiplicative behaviour that is tough to
	picture in ZX-calculus. \\
	As quantum circuits become arbitrarily large, reasoning on them is harder since the different rules may not scale, for example, in ZH-calculus. To overcome this,
	large-scale diagrammatic tools have been developed, such as
	!-boxes \cite{merry2014phd,quick2015firstorder} or scalable calculus \cite{carette2019szx}.
	These two might look very different: the first keeps
	the syntax of the base graphical calculus where some subdiagrams
	can be repeated any number of times to model a family of diagrams;
	the other adds new generators to picture any number of wires as
	a single one. Our goal is to show that the two are very much alike
	and that they have a similar expressivity. A translation between them
	also allows to write diagrams using both, and would give
	a new understanding of different works done in one or the other.
	
	\textbf{Contribution:} In this paper, we compare and translate !-boxes and scalable notations on several relevant examples. First, we consider the matrix arrows generators
	of scalable ZH-calculus. Translating theme into !-boxes gives a practical translation between the frameworks and allows direct diagrammatical proofs of arrows properties dramatically shorter than the proofs by induction of \cite{carette2019szx}. Then, the transformations expressed in \cite{lemonnier2021hypergraph} using !-boxes, namely hyper-local complementation and Fourier hyper pivot, are here pictured in scalable notations. We provide concise, scalable proofs of those transformations. However, in doing so we lose the readability of the topological structure allowed by !-boxes. Finally, by exhibiting the connection between matrix arrows and the trapezes of \cite{kuijpers2019graphical}, we show that scalable notations are also well-suited for graphical transforms. Relying on this, we give alternative proofs of the spider nest identities of \cite{munsonand,de2020fast} by direct transform computation.
	
	\textbf{Structure of the paper:} We first ZH-calculus is introduced in its well-tempered form \cite{de2020well} in Section \ref{sec:zh}. Once both frameworks, namely !-boxes and SZH-calculus, are defined in Section \ref{sec:large}, we describe heuristics to translate between the two in section \ref{sec:dico}. We then give concrete examples of applications of those heuristics to the hyper local complementation and regular hyper pivot \cite{lemonnier2021hypergraph}. Finally, in Section \ref{sec:transform}, we introduce graphical transforms in SZH-calculus and use them to derive the Fourier hyper pivot \cite{lemonnier2021hypergraph} and spider nest identities \cite{munsonand,de2020fast}.
	
	\section{ZH-calculus}\label{sec:zh}
	
	The ZH-calculus is a graphical language introduced by Backens and Kissinger~\cite{backens2018zh}
	to represent linear maps in a complex vector space through diagrams -- which will will be 
	referred to as \textit{ZH-diagrams}. A ZH-diagram is a collection of generators and wires
	between them. When $a$ is a complex number, the generators are:
	$$\input{figures/h-basic.tikz}\qquad\input{figures/z-basic.tikz}$$
	In colloquial speech, these can be called 
	respectively \textit{h-box} and \textit{green dot}, or more vaguely \textit{spiders}.
	The generators are given a semantics in complex matrices spaces, which
	is in itself compositional and provides then a semantics for any ZH-diagram.
	We consider in this paper the well-tempered~\cite{de2021well} semantics:
	$$\sem{\input{figures/z-basic.tikz}}=2^{\frac{m+n-2}{4}}\left(\ket 0^{\otimes n}\bra 0^{\otimes m}
	+\ket 1^{\otimes n}\bra 1^{\otimes m}\right) \quad
	\sem{\input{figures/h-basic.tikz}}=2^{\frac{-m-n}{4}}\sum_{x,y\in\{0,1\}^{m+n}}
	a^{x_1\dots x_my_1\dots y_n}\ket y\bra x$$
	Diagrams can be composed horizontally or vertically, meaning:
	$$\sem{\input{figures/diag-tensor.tikz}}
	= \sem{\input{figures/diag.tikz}}\otimes \sem{\input{figures/diag-prime.tikz}}
	\qquad\qquad 
	\sem{\input{figures/diag-composition.tikz}}
	=\sem{\input{figures/diag.tikz}} \circ \sem{\input{figures/diag-prime.tikz}}$$
	The tensor product is symmetric, meaning that we allow swapping wires, some
	natural equations in the diagrammatic calculus arise:
	$$\input{figures/eq-swap.tikz}$$
	Loose ends on the left and on the right are respectively inputs and outputs
	of the diagram. 
	It follows that a single wire with no generator performs
	the identity. Besides, we allow wires to bend:
	$$\sem{\input{figures/id.tikz}} = \ketbra 0 + \ketbra 1 \qquad
	\sem{\input{figures/cup.tikz}} = \ket{00}+\ket{11} \qquad
	\sem{\input{figures/cap.tikz}} = \bra{00}+\bra{11}$$
	The choice of well-tempered scalars is significant: some equations stated
	in the following sections have application in the verification of 
	quantum circuits, for example, hence the necessity to keep
	the scalars -- which is not always the case in graphical calculi. Another useful spider is the X-spider -- or \textit{red dot} --
	which is a generator from ZX-calculus and can be defined as a ZH diagram below.
	An empty h-box is actually an h-box with phase $-1$, such that a two-wired
	empty h-box pictures exactly the Hadamard gate.
	The diagrammatic equivalent of the $\texttt{NOT}$ gate is defined below.
	$$\input{figures/def-x-basic.tikz}\qquad\qquad\input{figures/def-not-basic.tikz}$$
	
	Diagrams are drawn until now with wires from left to right to easily understand
	them as inputs and outputs w.r.t. the semantics, but they can as well be written
	from bottom to top. Besides, only topology matters:
	$$\input{figures/bend-on-z.tikz}
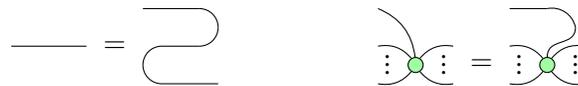
$$
	\begin{figure}
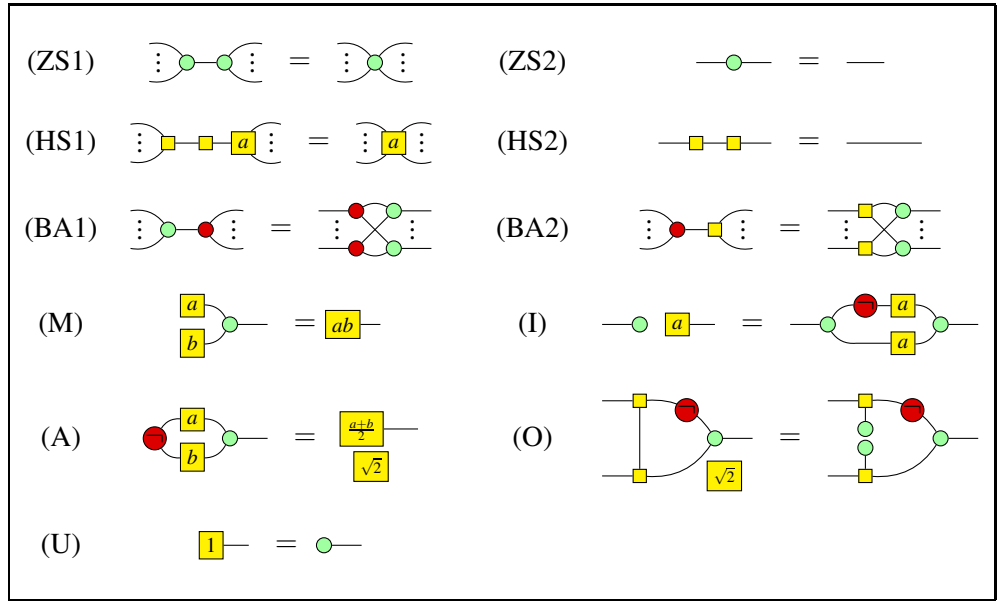

		\centering
		\begin{tabular}{|ccccc|}
			\hline
			&&&& \\
			(ZS1) & $\input{figures/zs1.tikz}$ & \qquad
			& (ZS2) & $\input{figures/zs2.tikz}$ \\ &&&& \\
			(HS1) & $\input{figures/hs1.tikz}$ & \qquad
			& (HS2) & $\input{figures/hs2.tikz}$ \\ &&&& \\
			(BA1) & $\input{figures/ba1.tikz}$ & \qquad
			& (BA2) & $\input{figures/ba2.tikz}$ \\ &&&& \\
			(M) & $\input{figures/m-axiom.tikz}$ & \qquad
			& (I) & $\input{figures/i-axiom.tikz}$ \\ &&&& \\
			(A) & $\input{figures/a-axiom.tikz}$ & \qquad
			& (O) & $\input{figures/o-axiom.tikz}$ \\ &&&& \\
			(U) & $\input{figures/u-axiom.tikz}$ &&& \\
			&&&& \\
			\hline
		\end{tabular}
		\caption{Rules of the ZH-calculus \cite{backens2018zh}}\label{fig:ZH-rules}
	\end{figure}
	
	ZH-calculus comes with a set of rules (see Fig.~\ref{fig:ZH-rules}) that is \textit{complete}, meaning that if two ZH-diagrams represent the same linear map, then those diagrams can be transformed into one another using these rules. There is no limit to how big a ZH-diagram can be; this paper takes a particular interest in two large-scale diagrammatic tools that allow to represent and manipulate unboundedly large diagrams.
	
	\section{Large-scale methods}\label{sec:large}
	
	This section introduces separately two large-scale diagrammatic tools, 
	namely !-boxes and scalable notations; to study and compare
	them in the oncoming sections.
	
	\subsection{!-boxes}
	
	Many equations of diagrams benefit from being written with notation
	that capture large-scale designs. Already some equations in Fig.\ref{fig:ZH-rules}
	need dots "$\dots$" to fully picture their behavior. However the semantics of 
	the dots is not well defined and their expression capability is very limited;
	they may as well become confusing. !-boxes -- read "bang-boxes" -- are an attempt at picturing large-scale diagrams more thoroughly. A !-box circling a part of a ZH-diagram means that this subdiagram	can be replicated any number of times; a diagram containing a !-box represents then a family of diagrams, as follows:
	$$\input{figures/def-bang.tikz}\quad =\quad 
	\left\{\quad\input{figures/def-bang-0.tikz}\quad ;\quad
	\input{figures/def-bang-1.tikz}\quad ;\quad
	\input{figures/def-bang-2.tikz}\quad ;\quad
	\input{figures/def-bang-3.tikz}\quad ;\quad \dots\quad
	\right\}$$
	If one requires to specify a diagram from the family, the !-box
	can be annotated, and we get the diagram on the left below.
	!-boxes can also be overlapped, resulting in a fully-connected bipartite graph on the right:
	$$\input{figures/def-bang-anno.tikz}
	\qquad\qquad\input{figures/def-bang-overl.tikz}$$
	A particular equation that requires graphical Fourier transform~\cite{kuijpers2019graphical}
	is pictured through !-boxes annotated with sets and parameters
	that cover the sets, like in the upcoming example.
	From now on, for any natural number $n$ we will write $[n]$ to denote the set 
	$\{0,\dots,n-1\}$.
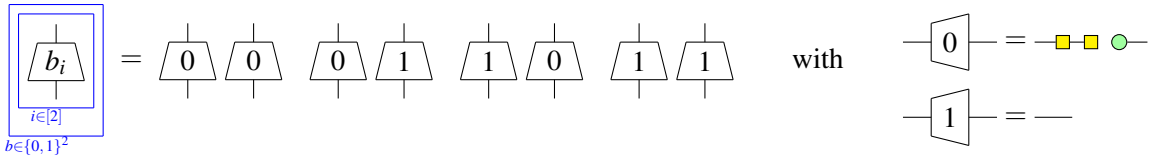
	$$\input{figures/def-trapeze-bang.tikz}\qquad\text{with}\qquad\input{figures/def-trapeze.tikz}$$
	We will also use gray trapezes:	$\qquad\input{figures/def-trapeze2.tikz}$.\\
	
	As an example of use, some rules of Fig.\ref{fig:ZH-rules} are better written with !-boxes:
	\begin{align*}
		\input{figures/zs1-bang.tikz} \qquad
		\input{figures/hs1-bang.tikz} \\
		\input{figures/ba1-bang.tikz} \qquad
		\input{figures/ba2-bang.tikz}
	\end{align*}
	
	\subsection{Scalable notations}
	
	The scalable notation have first been introduced in \cite{chancellor2016graphical} before being formalised in \cite{carette2019szx}. As opposed to !-boxes, the idea is to stay inside the prop framework by introducing new kinds of wires and generators representing large-scale processes.
	
	\subsubsection{SZH}
	
	In SZH, for each $n\in \mathbb{N}^* $, where $\mathbb{N}^*$ are the natural integers without zero, we have a type of wire of size $n$ denoted $(n)$. Formally the corresponding categorical structure will then be a $\mathbb{N}^*$-colored prop and the generic types are of the form $(n_1) +\cdots +(n_i )+ \cdots + (n_m )$. The \textbf{size} of a type is inductively defined as: $|(n)|\df n$ and $|a + b|\df|a|+ |b|$. We use thin wires to denote wires of type $(1)$ and thick wires denote any $(n)$. We write $(m)^n$ for the tensor product of $n$ wires $(m)$ with the convention $(m)^0=(0)$. We have two new families of generators, the dividers and the gatherers, for each positive integer $n$ satisfying the following equations:
	
	\begin{center}
		$\input{dd.tikz}$ \scalebox{0.8}{$: (n+1)\to (1)+ (n)$} $\qquad$ $\input{gg.tikz}$\scalebox{0.8}{$: (1)+ (n)\to (n+1)$} $\qquad$ $\input{ex0.tikz}=\input{ex1.tikz}\qquad\input{el0.tikz}=\input{el1.tikz}$
	\end{center}
	
	For every generator $n\to m$ of ZH and every $k\in \mathbb{N}^*$ there is a generator $(k)^n \to (k)^m$ in SZH. Those scaled generators satisfy the same equations as the original ones:
	
	\begin{center}
		$\input{z.tikz}:(k)^n\to (k)^m \qquad \input{h.tikz}:(k)^n\to (k)^m$
	\end{center}
	
	They interacting with dividers and gatherers as:
	
	\begin{center}
		$\input{dig0.tikz}=\input{dig1.tikz}\qquad\input{dih0.tikz}=\input{dih1.tikz}$
	\end{center}
	
	Note that here $\mathbf{a}\in \mathbb{C}^k $ is a list of $k$ complex parameters .$\mathbf{a}_0 $ is the head of $\mathbf{a}$ and $\mathbf{a}' $ its tail. We can generalize the following rules by taking addition and multiplication of lists pointwise:
	
	\begin{center}
		\input{genrules.tikz}
	\end{center}
	
	It is shown in \cite{carette2020colored} that any diagram in SZH can be rewritten into normal form.
	
	\begin{lemma}[Normal form]
		For each $f:a\to b$ in SZH there is a unique $|f|:|a|\to |b|$ in ZH such that:
		
		\begin{center}
			\input{normalform.tikz}
		\end{center}
		
	\end{lemma}
	
	This allows to extend the size map into the \textbf{wire stripping} functor $|\_|:SZH\to ZH$. We can then extend the interpretation of ZH into an interpretation of SZH, $\interp{\_}:SZH\to \mathbf{FHilb}_2$, defined as $\interp{D}=\interp{|D|}$. We have $\interp{(n)}\df\mathbb{C}^{2^n }$ and $\interp{a + b}\df \interp{a}\otimes\interp{b}$. The wires have interpretations:
	
	\begin{center}
		$\interp{\input{n.tikz}}\df \sum\limits_{x\in \textbf{2}^n } \ket{x}\bra{x} \quad\interp{\input{swap.tikz}}\df \sum\limits_{x,y\in \textbf{2}^n } \ket{x}\ket{y}\bra{x}\bra{y}\quad\interp{\input{bcup.tikz}}\df \sum\limits_{x\in \textbf{2}^n}\ket{x}\ket{x}\quad\interp{\input{bcap.tikz}}\df \sum\limits_{x\in \textbf{2}^n}\bra{x}\bra{x}$.
	\end{center}
	
	The dividers and gatherers act trivially:
	$\interp{\input{dd.tikz}}\df \sum\limits_{x\in \textbf{2}^{n+1} } \ket{x}\bra{x}$ and $\interp{\input{gg.tikz}}\df \sum\limits_{x\in \textbf{2}^{n+1} } \ket{x}\bra{x}$.\\
	
	The arachnids of type $(k)^n\to (k)^m$ have interpretations:
	
	\begin{center}
		$\interp{\input{z.tikz}}\df 2^{k\frac{n+m-2}{4}}\sum\limits_{x\in {\bf 2}^k} e^{i(x\cdot a)}\ket{x}^{\otimes n}\!\bra{x}^{\otimes m}\qquad\interp{\input{h.tikz}}\df 2^{-k\frac{n+m}{4}}\sum\limits_{x_i\in {\bf 2}^k} \prod\limits_{j=1}^{k} a_j^{\bigwedge\limits_{i=1}^{n+m} x_{i,j}}\ket{x_1 \cdots x_n}\!\bra{x_{n+1} \cdots x_{n+m}}$
	\end{center}
	Where the $x_i$ are binary words of size $k$ and $x_{i,j}$ is the $j$-th bit of $x_i$.
	
	\subsubsection{Matrix arrows}
	
	The main interest of scalable notations for ZH is the manipulation of matrix arrows, new kinds of generators encapsulating bipartite graph diagrams that were already present in \cite{chancellor2016graphical}. There are two kinds, the red and the yellow ones:
	
	\begin{center}
		\input{defarrows.tikz}
	\end{center}
	
	Where $A\in \mathcal{M}_{m\times n}(\{0,1\})$ is a $\{0,1\}$-matrix corresponding to the bi-adjacency matrix of the green/red and green/yellow bipartite graphs respectively. In particular:
	
	\begin{center}
		\input{scalarmat.tikz}
	\end{center}
	
	By convention if the matrix is not specified then it is considered to be full of $1$. Arrows are copied and erased by green nodes and co-copied and co-erased respectively by red and yellow nodes:
	
	\begin{center}
		\input{copiesarrows.tikz}
	\end{center}
	
	They can be composed in the following ways:
	
	\begin{center}
		\input{composearrows.tikz}
	\end{center}
	
	Where $BA$ is the matrix product taken over the field $\mathbb{F}_2 $ and $B\cdot A $ is the matrix product taken over the boolean ring $\mathbb{B}$ with $+ = \lor$ and $\times = \land $. Note that this last product is, in practice, easier to compute than the first one. The interpretation is then:
	
	\begin{center}
		$\interp{\input{rarrow.tikz}}= \ket{x}\mapsto 2^{\frac{m-n}{4}}\ket{Ax}\qquad\interp{\input{yarrow.tikz}}= \ket{x}\mapsto 2^{\frac{m-n}{4}}\ket{A\cdot x}$
	\end{center}
	
	Finally, there are two properties that are very useful in practice, the Hadamard gate can flip red arrows and when the matrix $C$ has at most one $1$ in each row then red and yellow arrows are equivalent:
	
	\begin{center}
		\input{techarrows.tikz}
	\end{center}
	
	\section{A dictionary between !-boxes and scalable notations}\label{sec:dico}
	
	Now that the two styles of large-scale diagrammatic reasoning tools have been introduced, we will investigate their respective expressivity and efficiency by considering the translation from one formalism to the other in various cases. Here we will denote $\leftrightarrow$ to express that a scalable and a !-box diagram correspond to the same ZH diagram. A rigorous proof that !-boxes and scalable notations have the same expressiveness would require the development of a framework able to handle parametrized families of diagrams. Some work has been done on this matter~\cite{kissinger2015banggraphs,kissinger2015contextfree,zamdzhiev2019framework,quick2015firstorder,merry2014phd}, but they do not picture !-boxes in a compositional way and often do not encompass overlapping. We do not ambition to do it here. However, we provide translation heuristics that strongly support the conjecture that the expressiveness of the two approaches is indeed similar.
	
	\subsection{Boxes and arrows}\label{ssection:boxesarrows}
	
	The scaled generators have direct corresponding !-box diagrams up to permutations of the inputs and outputs:
	
	\begin{center}
		\input{scalgenbang.tikz}
	\end{center}
	
	An important remark that will be at the core of all our examples is that the matrix arrows are in fact very nicely expressed as !-boxes:
	
	\begin{center}
		\input{bagarrows.tikz}
	\end{center}
	
	Knowing this one can for example easily show that matrix arrows are erased in one line of !-box manipulation:
	
	\begin{center}
		\input{bangmatcomp.tikz}
	\end{center}
	
	The compactness and elegance of this proof must be compared with the lengthy proofs by induction of \cite{carette2019szx}, typical of scalable notations. 
	Very similar proofs can be done for all properties of matrix arrows. We see that !-boxes allow for zooming in arrows to exploit their internal structure. In the other direction, it is now clear the trapezes are just arrows in disguise:
	
	\begin{center}
		\input{trapezearrow.tikz}
	\end{center}
	
	It is particularly striking that this structure appears to have been discovered independently for completely different reasons in \cite{chancellor2016graphical} and \cite{kuijpers2019graphical}. This correspondence between arrows and !-boxes is particularly fruitful in the following simple cases:
	
	\begin{center}
		\input{bagarrows2.tikz}
	\end{center}
	
	Those expressions will be instrumental to reformulate !-box diagrams into similar scalable ones. The general heuristic goes as follows: from scalable to !-boxes, expand the scaled generators and arrows, and from !-boxes to scalable, identify the relevant structures appearing in the dictionary and replace them with arrows.
	
	\subsection{(Hyper)local complementation}
	
	Graph states form a particular family of states defined as follows: starting from a graph on $n$ vertices, we take a tensor product of $n$ $\ket{+}$ states, corresponding to green nodes. Then we apply $C$-$Z$ gates for each edge of the graph two the corresponding nodes. $C$-$Z$ gates are represented as: \input{czgate.tikz}, any graph state can then be represented as: \input{graphstate.tikz}. Where $K_n $ is a $\binom{n}{2}\times n $-matrix defined as $(K_n )_{v,e}= v\in e $, and $G$ is a $\{-1, 1\}$-vector that characterizes the edges in $G$ by $-1$ and the edges not in $G$ by $1$. Graph-states enjoy numerous interesting properties, and the interested reader can consult for example \cite{hein2006entanglement} for more details. Here we will focus on the local complementation rule, that can be expressed as:
	
	\begin{center}
		\input{loccomp.tikz}
	\end{center}
	
	This rule is the core of the following identity. 
	
	\begin{theorem}[Hyper local complementation \cite{lemonnier2021hypergraph}] We have: $\quad\input{figures/hlc.tikz}$
	\end{theorem}
	
	Hyperlocal complementation cannot be fully pictured without
	using dots, even with !-boxes. Dots could be avoided by creating
	some new notation or a new definition for !-boxes, but it would not
	fit naturally in the current diagrammatic language. A similar theorem can be formulated in scalable notations, with the following correspondence:
	$$\input{figures/hlc-explain.tikz}
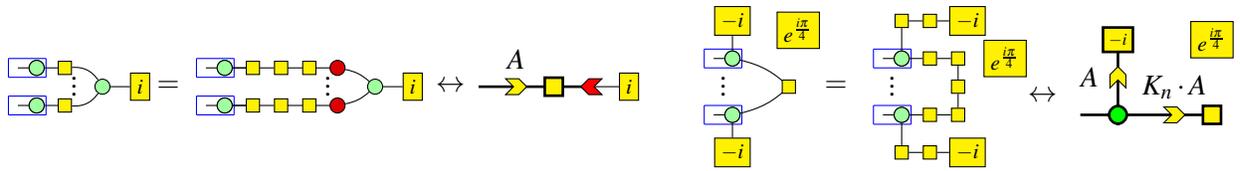
$$
	
	\begin{theorem}[Hyper local complementation, scalable version]
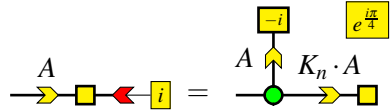
\label{thm:scalhypercomp}
		$$\input{scalhypercomp.tikz}$$
	\end{theorem}
	
	Here the scalable version is more general and less precise than the !-box version. Here the !-boxes exhibit an internal structure that is hidden inside the matrix arrows.
	
	\subsection{Regular hyper pivot}
	
	The regular hyper pivot is a generalisation of the pivot rule.
	The latter was
	formalized for ZX-calculus \cite{duncan2009graph,pivoting} in order to 
	enhance quantum circuits optimisation. The hyper pivot rule arises from
	applying the same graph-theoretic transformation as pivoting, but in a
	hypergraph.
	\begin{theorem}[Regular hyper pivot \cite{lemonnier2021hypergraph}]\label{th:rhp}
		The following equation holds in ZH-calculus:
		$$\input{figures/rhp.tikz}$$
	\end{theorem}
	It is interesting to note that (HS2), (BA1) and (BA2) from Fig.\ref{fig:ZH-rules} are
	particular cases of the regular hyper pivot, as proven in~\cite{lemonnier2021hypergraph}.
	The regular hyper pivot can easily be translated into SZH-calculus, considering the
	correspondance shown in Section~\ref{ssection:boxesarrows}:
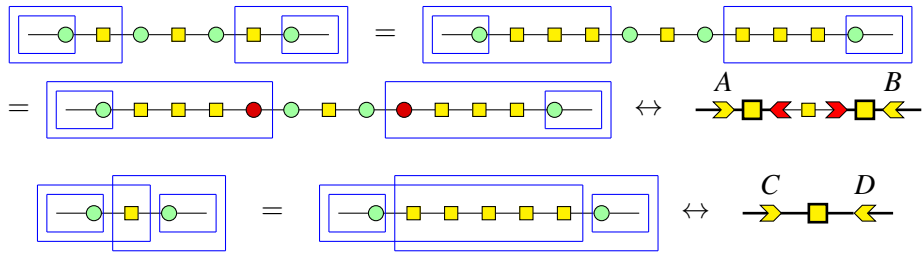
	$$\input{figures/rhp-explain.tikz}$$
	It gives the following theorem, proven in SZH-calculus.
	\begin{theorem}[Regular hyper pivot, scalable version]
		The following equation holds in SZH-calculus, for any Boolean matrices $A,B$:
		$$\input{figures/rhp-scal.tikz}$$
	\end{theorem}
	The proof is very short:
	$$\input{figures/rhp-scal-proof.tikz}$$
	The second to last equality is given by the following lemma:
	\begin{lemma}
		This equation holds in SZH-calculus, for any Boolean matrices $A,B$:
		$$\input{figures/rhp-lemma.tikz}$$
	\end{lemma}
	Besides, this proof is very similar to the one done with !-boxes: there is no
	fundamental difference between !-boxes and scalable notation regarding the regular
	hyper pivot. However, the following section tackles a set of diagrams that seem more
	natural written in SZH-calculus.
	
	\section{Diagrammatic transforms}\label{sec:transform}
	
	Diagrammatic transforms have been introduced in \cite{kuijpers2019graphical} as an efficient way to transform ZX diagrams into ZH diagrams and \textit{vice versa}. In this section, we present a different approach to diagrammatic transforms that allows us to consider them either with !-boxes or with scalable notations.
	
	\subsection{Transforms of boolean functions}
	
	An invertible linear map $D:2^n \to 2^n $ which is diagonal in the computational basis can be described by a boolean function $f: \{0,1\}^n \to \mathbb{C}^*$ such that $D\ket{x}=f(x)\ket{x}$. The pointwise multiplication of boolean functions $(fg)(x)=f(x)g(x)$ then corresponds to the composition of the underlying diagonal linear maps. We will write $f^\lambda (x)= f(x)^\lambda$ for any $\lambda\in \mathbb{Z}$ with the convention that $z^0 =1$. We call \textbf{transform} a group automorphisms of ${\mathbb{C}^*}^{\{0,1\}^n}$ that is, a map $T:\mathbb{C}^{\{0,1\}^n} \to \mathbb{C}^{\{0,1\}^n} $ such that $T(fg)=T(f)T(g)$ and $T(f^\lambda )= T(f)^\lambda $. We will mainly consider two such transforms, the Fourrier and the Mobius transforms:
	
	\begin{lemma}\label{lm:transformsBool}
		The Fourrier  and Mobius transforms $\hat{f}$ and $\tilde{f}$ defined respectively as:
		
		\begin{center}
			$ \hat{f}(s)=\prod\limits_{t} f(t)^{-2^{1-n}(-1)^{s\cdot t}}\quad$ and $\quad \tilde{f}(s)\doteq\prod\limits_{t\leq s} f(t)^{(-1)^{|t|+|s|}}$ 
		\end{center}
		
		satisfies:
		
		\begin{center}
			$f(x)=f(\emptyset)\prod\limits_{s} \check{f}(s)^{s\cdot x }\quad$ and $\quad f(x)=\prod\limits_{s\leq x} \tilde{f}(s)$.
		\end{center}
		
	\end{lemma}

	Where $|x|$ denotes the Hamming weight of $x$, which is the number of one in it, and $s\cdot t \doteq \bigoplus\limits_{i=0}^{2^n -1} x_i y_i$. One needs to be careful here, the power $2^{1-n}$ is no uniquely defined if $ 1< n$, in fact the decomposition is not unique and choosing any solution leads to a valid decomposition. More precisely for each $t$ we can choose any of the $2^{n-1)} $ solutions of the equation $X^{2^{n-1}}=f(t)$, this gives $2^{2n-1}$ possibilities for the Fourrier transform. However, all of those possibilities satisfies the inversion formula.
	
	\subsection{Transforms of diagram}
	
	The practical interest of those transforms is that they can be interpreted as a way to realize any diagonal map $D: 2^n \to 2^m $ as a composite of elementary gates. Given $s \subseteq \{1,\dots,n\}$ and $\lambda \in \mathbb{C}^* $ we define the \textbf{phase-gadget} $X_{s,\lambda}$ and the \textbf{generalized hyper-edge} $H_{s,\lambda}$ respectively as:
	
	\begin{center}
		$X_{s,\lambda} \ket{x}\doteq \lambda^{s\cdot x} \ket{x} \qquad \interp{\input{phasegadget.tikz}}= X_{s,\lambda}\qquad
		H_{s,\lambda} \ket{x}\doteq \lambda^{\delta_{s\leq x}}\ket{x} \qquad \interp{\input{hyperedge.tikz}}= H_{s,\lambda}$
	\end{center}
	
	We see that a composition of phase-gadgets or hyper-edges can be rewritten using arrows. Thus we have:
	
	\begin{lemma}\label{lm:mobiusFourrier}
		Given a diagonal linear map $D:2^n \to 2^n $ and the corresponding boolean function $f: \{0,1\}^n \to \mathbb{C}^* $ we have:
		
		\begin{center}
			$\interp{\input{fourrierdiag.tikz}}=\interp{\input{mobiusdiag.tikz}}= D$, $\quad\input{fourrierdiag2.tikz}=\input{mobiusdiag2.tikz}\quad$ and \begin{tabular}{l}
				$\tilde{f}(x)=\prod\limits_{t\geq x}  \hat{f}(t)^{-2^{|x|-1}}$ \\
				$\hat{f}(x)=\prod\limits_{t\geq x}  \tilde{f}(t)^{-2^{1-|t|}(-1)^{|x|}}$\\
				$f(\emptyset)=\tilde{f}(\emptyset)= \prod\limits_{t}  \hat{f}(t)^{-\frac{1}{2}}$
			\end{tabular} 
		\end{center}
		
	\end{lemma}
	
	Where the $2^n \times n $-matrix $\Delta_n$ is defined as $(\Delta_n)_{i,s} \df \delta_{i\in s}$. This matrix contains all possible positions for the legs of hyper-edges and phase gadgets. This happens to be very useful in practice when we want to switch between weighted hyper-graph and phase-gadget nests. The matrix $S_n $ here is just a reformulation of the trapezes. What we see here is just a direct reformulation of the original formulation of diagrammatic transform introduced in \cite{kuijpers2019graphical}.
	
	\begin{center}
		\input{transform.tikz}
	\end{center}
	
	\subsection{Spider nest identities}
	
	Spider nest identities have been introduced in \cite{de2020fast} useful rewriting rules to simplify quantum circuits. Formally, a spider nest identity is given by a composition of spider-like diagrams, typically generalized hyper-edges, and phase gadgets, with one big spider and a lot of tiny ones, which equals the identity. Those identities can be derived directly using the diagrammatic transforms. We here consider the special case of \textbf{symmetric boolean functions}, that is, when $f(x)$ only depends on the Hamming weight of $x$. We write $F:\textbf{n}\to \mathbb{R}$ the function such that $f(x)=F(|x|)$. We then have more compact transforms.
	
	\begin{lemma}\label{lm:transformsSymBool}
		The Kravchuk and Binomial transforms $\hat{F}$ and $\tilde{F}$ defined respectively as:
		
		\begin{center}
			$\hat{F}(m)\df \prod\limits_{k=0}^{n}F(k)^{-2^{1-n}K^n_k(m)}$
			and
			$\tilde{F}(m)\df\prod\limits_{k=0}^{m} F(k)^{\binom{m}{k} (-1)^{m-k}}$
		\end{center}
		
		where $K^n_k$ is the \textbf{Kravchuk polynomial}: $K^n_k(X)\df\sum\limits_{j=0}^{k}\binom{X}{j}\binom{n-X}{k-j}(-1)^j$, satisfy:
		
		\begin{center}
			$\hat{f}(x)=\hat{F}(|x|)$, $\quad\tilde{f}(x)=\tilde{F}(|x|)$,  $\quad F(m)=\prod\limits_{k=0}^n \hat{F}(k)^{\frac{-1}{2}K^n_k(m)}\quad$ and
			$\quad F(m)=\prod\limits_{k=0}^{m}\tilde{F}(k)^{\binom{m}{k}}$
		\end{center}
		
	\end{lemma}
	
	See \cite{canteaut2005symmetric} for details on transforms of symmetric boolean functions. Note the $\frac{1}{2^n}$ in exponent for the Kravchuk transform, this means that the previous remark on the uniqueness of the Fourier transform also applies here. Direct computations using those transforms allow to derive the two following spider-nest identities:
	
	\begin{lemma}\label{lm:nest} We have:
			$\quad\input{mobiusgadget.tikz}\quad$ and $\quad\input{spidernest.tikz}$\\
		
		where $\omega=e^{\frac{i\pi}{4}}$ and $\hat{s}$ is a symmetric boolean function defined as:
		
		\begin{center}
			$\hat{S}(0)=1$ $\hat{S}(1)=e^{\frac{i\pi(n-2)(n-3)}{16}}$ $\hat{S}(2)=e^{\frac{-i\pi(n-3)}{8}}$ $\hat{S}(3)=e^{\frac{i\pi}{8}}$ $\hat{S}(n)=e^{\frac{-i\pi}{8}}$  $\hat{S}(k)=1$ for $ k\neq 0,1,2,3,n$.
		\end{center}
		
	\end{lemma}
	
	The first one has been derived using a proof by induction in \cite{munsonand}, the second one introduced in \cite{de2020fast} was first derived using the diagrammatic transform from \cite{} and successfully used to simplify quantum circuits. The method used to derive those equalities open the way to the automatic mining of spider-nest identities.
	
	\subsection{Fourier hyper pivot}
	
	The Fourier hyper pivot is a diagrammatic equation inspired by one of
	the rewrite rules created by Amy~\cite{amy2019verification}. Those rewrite
	rules are not between diagrams but algebraic expressions, nevertheless both 
	frameworks carry the same expressiveness, as stated in~\cite{vilmart2021sop,lemonnier2021hypergraph}. Amy generalized some quantum circuit equations to perform efficient verification; and the rule which came out of the Hadamard boxes cancellation is very successful on that matter. The Fourier hyper pivot is the direct translation of the latter. This diagrammatic equation is, in addition, a relevant example of the use of the graphical Fourier transform~\cite{kuijpers2019graphical} and it motivates the use of !-boxes
	instead of dots in ZH-calculus.
	
	\begin{theorem}[Fourier hyper pivot~\cite{lemonnier2021hypergraph}]
		The following equation holds for any $n,m\in\N$, any family
		of complex numbers $\lambda$ and where $\vert b\vert$ denotes
		the weight of $b$:
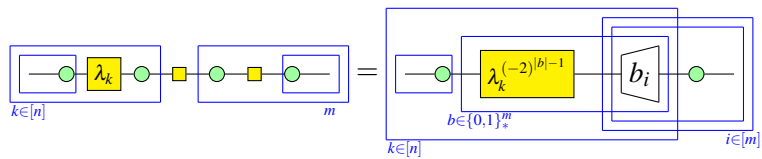
		$$\input{figures/fhp-inside.tikz}$$
	\end{theorem}
	
	The regular hyper pivot is a particular case of this equation. It can be seen
	rather easily by replacing all the $\lambda_i$ with $-1$. SZH-calculus seems more adapted to picture the Fourier hyper pivot: a family of trapezes is directly an arrow, and we will see that it is a generalisation of Theorem~\ref{th:rhp}.
	
	\begin{theorem}[Fourier hyper pivot, scalable version]\label{thm:fhp}
		The following equation holds in SZH-calculus, for any Boolean
		matrices $A,B$ and complex vector $\fl\lambda$:
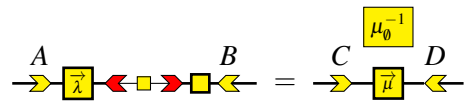
		$$\input{fhp-scal.tikz}$$
		where $C= \Delta_{nm}\cdot \overrightarrow{I_n} \cdot A$, $D=\Delta_{nm}\cdot \overrightarrow{I_m} \cdot B$ and $\mu_{\emptyset}=\prod\limits_{i} \lambda_i^{\frac{-1}{2}}$, $\mu_s= \begin{cases}
			\lambda_i^{-2^{|s|-1}} \text{ if } s\leq \overrightarrow{\{i\}}\\
			1 \text{ else}
		\end{cases}$ where $s\in \{0,1\}^{nm}$, and $\overrightarrow{\{i\}}$ is the binary words of $m$ copies of the binary word of size $n$ with only one $1$ at position $i$.
	\end{theorem}
	
	One can see here that the whole complexity of the expressions is almost hidden and lies in the matrices and the vectors that index the diagrams. 
	
	\section{Conclusion}
	
	We have studied the interplay of !-boxes and scalable notations in many examples of large-scale graphical reasoning. It appears that the indexed trapeze box of \cite{kuijpers2019graphical} and \cite{lemonnier2021hypergraph} to present diagrammatic transforms exactly correspond to the arrows of scalable notations. Thus, all reasoning involving matrix arrows in scalable notations, for example in \cite{carette2021quantum}, can be straightforwardly carried out with !-boxes. Using this correspondence, we express graphical transforms in a very compact way using matrix arrows, leading to shorter proofs than the ones given in \cite{kuijpers2019graphical}. This new point of view allowed us to provide direct derivations of the spider nest rules of \cite{de2020fast,munsonand}. Using a hypergraph inspired representation of graph-sates we also dramatically simplified the diagrammatical representation of local complementation from \cite{carette2019szx}. Turning !-boxes into scalable notations happens to be not so challenging given the proper heuristics. In general, finding a scalable diagram with the right shape is easy. The difficulty is in finding the adequate matrices to label the arrows. We experimented with this on the Fourier Hyper pivot of \cite{lemonnier2021hypergraph}. Relying on our new formulation of diagrammatic transforms and local complementation, we managed to provide alternative shorter proofs for most results from \cite{lemonnier2021hypergraph}. This translation exercise made it clear that what scalable notations win in compactness, they lose in faithfulness in exhibiting the internal structure of the diagram. While !-boxes, on the contrary, provide an interesting account of those medium-scale topological structures. To sum up, we conclude by emphasizing three key messages to take from those investigations:
	
	\begin{itemize}
		
		\item Large-scale diagrammatic reasoning tools provide powerful rewriting techniques to tidy up cumbersome diagrammatical proofs by exhibiting the relevant topological structure involved and abstracting away diagrammatical noise. Those tools form a bridge between low and higher-level reasoning on quantum processes.
		
		\item Given the translation heuristics we presented, it seems highly plausible that !-boxes and scalable notations have the same expressiveness. However, a proper formulation and proof of this conjecture requires developing a framework for parametrized families of diagrams able to encompass both models. Such a framework is also the condition for a hypothetical safe hybrid usage of both tools at once.
		
		\item The essential difference between the two frameworks is then at the level of practical efficiency. When to prefer one tool over the other is mainly a question of scale and coarse-graining. The !-boxes have the ability to exhibit the medium-scale topological structure of diagrams, while the scalable notations tend to encapsulate everything into very compact matrix arrows, the medium-scale topological structure being hidden in the indexing matrices.
		
	\end{itemize}
	
	\bibliographystyle{eptcs}
	\bibliography{scal}
	
	\appendix
	
	\section{Proofs}
	
	\begin{proof}[Proof of Theorem \ref{thm:scalhypercomp}]$\quad$\\
		
		\begin{center}
			\input{scalhypercomp2.tikz}
		\end{center}
		
	\end{proof}
	
	\begin{proof}[Proof of Lemma \ref{lm:transformsBool}]$\quad$\\
		
		First for the Fourier transform:\\
		
		\begin{center}
			$f(\emptyset)\prod\limits_{s} \hat{f}(s)^{s\cdot x }=f(\emptyset)\prod\limits_{s} \left(\prod\limits_{t} f(t)^{-2^{1-n}(-1)^{s\cdot t}}\right)^{s\cdot x }
			= f(\emptyset)\prod\limits_{t,s} f(t)^{-(s\cdot x)2^{1-n} (-1)^{s\cdot t}}
			= f(\emptyset)\prod\limits_{t}  f(t)^{-2^{1-n}\sum\limits_{s} (s\cdot x) (-1)^{s\cdot t}}$
		\end{center}
		
		We evaluate $2^{1-n}\sum\limits_{s} (s\cdot x) (-1)^{s\cdot t}$. Using $s\cdot x= \frac{1-(-1)^{s\cdot x}}{2}$ and $\sum\limits_{t} (-1)^{s\cdot t}=2^n \delta_{s= \emptyset}$ we have:
		
		\begin{center}
			$2^{1-n}\sum\limits_{s} (s\cdot x) (-1)^{s\cdot t}=2^{-n}\sum\limits_{s} (-1)^{s\cdot t} -2^{-n}\sum\limits_{s} (-1)^{s \cdot (x\oplus t)}=\delta_{t= \emptyset} -\delta_{x=t}$
		\end{center}
		
		So:
		
		\begin{center}
			$f(\emptyset)\prod\limits_{t} \hat{f}(t)^{s\cdot t }=f(\emptyset)\prod\limits_{t}  f(t)^{\delta_{x=t} -\delta_{t=\emptyset}}
			= f(\emptyset)f(x)f(\emptyset)^{-1}
			= f(x)$
		\end{center}

		For the Mobius transform:\\
		
		\begin{center}
			$\prod\limits_{s\leq x} \tilde{f}(s)= \prod\limits_{s\leq x} \prod\limits_{t\leq s} f(t)^{(-1)^{|t|+|s|}}
			= \prod\limits_{t\leq x} \quad \prod\limits_{t\leq s\leq x} f(t)^{(-1)^{|t|+|s|}}
			= \prod\limits_{t\leq x} f(t)^{\sum\limits_{t\leq s\leq x}(-1)^{|t|+|s|}}$
		\end{center}
		
		We evaluate $\sum\limits_{t\leq s\leq x}(-1)^{|t|+|s|}$. If $t=x$ then $ \sum\limits_{t\leq s\leq x}(-1)^{|t|+|s|}=1 $, else there is a position $i$ with $t_i =0$ and $x_i =1 $, without loss of generality we assume that $i=0$.
		
		\begin{center}
			$\sum\limits_{t\leq s\leq x}(-1)^{|t|+|s|}= \sum\limits_{t\leq 0s'\leq x}(-1)^{|t|+|s'|} + \sum\limits_{t\leq 1s'\leq x}(-1)^{|t|+1+|s'|} = 0$
		\end{center}
		
		So $\sum\limits_{t\leq s\leq x}(-1)^{|t|+|s|} = \delta_{t=x}$ and we have:
		
		\begin{center}
			$\prod\limits_{s\leq x} \tilde{f}(s)= \prod\limits_{t\leq x} f(t)^{\sum\limits_{t\leq s\leq x}(-1)^{|t|+|s|}} 
			= \prod\limits_{t\leq x} f(t)^{\delta_{t=x}}
			= f(x)$
		\end{center}
		
	\end{proof}
	
	\begin{proof}[Proof of Lemma \ref{lm:mobiusFourrier}]$\quad$\\
		
		The matrix $\Delta_n $ allows to compose all possible hyper-edges or phase-gadgets. The corresponding parameters $\lambda$s are given in the $H$-boxes. We then obtain diagonal gates that are compositions of elementary phase-gadgets and hyper-edges. This exactly corresponds to the inversion formula for the two transforms, hence the interpretation as the diagonal gate $D$. The equality of the two diagrams then follows by completeness. Now, for the formulas:\\
		
		\begin{center}
			$\tilde{f}(x)=\prod\limits_{s\leq x} f(s)^{(-1)^{|s|+|x|}} 
			= \prod\limits_{s\leq x} \left(f(\emptyset)\prod\limits_{t} \hat{f}(t)^{s\cdot t }\right)^{(-1)^{|s|+|x|}}
			= f(\emptyset)^{\sum\limits_{s\leq x} (-1)^{|s|+|x|}}\prod\limits_{t}  \hat{f}(t)^{\sum\limits_{s\leq x} (s\cdot t)(-1)^{|s|+|x|}}$
		\end{center}
		
		We evaluate $\sum\limits_{s\leq x} (-1)^{|s|+|x|}$. If $x=\emptyset$ then $\sum\limits_{s\leq x} (-1)^{|s|+|x|}= 1 $, else there is a position $i$ with $x_i =1$, without loss of generality we assume that $i=0$. Then:
		
		\begin{center}
			$\sum\limits_{s\leq x} (-1)^{|s|+|x|}=\sum\limits_{0s'\leq 1x'} (-1)^{|0s'|+|x|}+\sum\limits_{1s'\leq 1x'} (-1)^{|1s'|+|x|}= \sum\limits_{s'\leq x'} (-1)^{|s'|+|x|}-\sum\limits_{s'\leq x'} (-1)^{|s'|+|x|}=0$
		\end{center}
		
		Thus we have: $\sum\limits_{s\leq x} (-1)^{|s|+|x|}=\delta_{x=\emptyset}$.
		
		We now evaluate $\sum\limits_{s\leq x} (s\cdot t)(-1)^{|s|+|x|}=\frac{1}{2}\sum\limits_{s\leq x} (-1)^{|s|+|x|} - \frac{1}{2} \sum\limits_{s\leq x}(-1)^{s\cdot t + |s|+|x|}=\frac{1}{2}\delta_{x=\emptyset}-\frac{1}{2}\sum\limits_{s\leq x}(-1)^{s\cdot t + |s|+|x|}$.
		
		It remains to compute $\sum\limits_{s\leq x}(-1)^{s\cdot t + |s|+|x|}$. First if $x=\emptyset$ we get $ 1$. If $x\neq \emptyset$ and we don't have $x\leq t$ then there is a position $i$ such that $t_i = 0$ and $x_i =1 $, assuming without loss of generality that $i=0$ we get: 
		
		\begin{align*}
			\sum\limits_{s\leq x}(-1)^{s\cdot t + |s|+|x|}&= \sum\limits_{0s'\leq 1x'}(-1)^{(0s')\cdot (0t') + |0s'|+|1x'|}+\sum\limits_{1s'\leq 1x'}(-1)^{(1s')\cdot (0t') + |1s'|+|1x'|}\\
			&= \sum\limits_{s'\leq x'}(-1)^{s'\cdot t' + |s'|+|1x'|}-\sum\limits_{s'\leq x'}(-1)^{s'\cdot t' + |s'|+|1x'|}\\ &=0
		\end{align*}
		
		Else if we assume, without loss of generality, that $t_0 = 1$ and $x_0 =1$ we get:
		
		\begin{align*}
			\sum\limits_{s\leq x}(-1)^{s\cdot t + |s|+|x|}&= \sum\limits_{0s'\leq 1x'}(-1)^{(0s')\cdot (1t') + |0s'|+|1x'|}+\sum\limits_{1s'\leq 1x'}(-1)^{(1s')\cdot (1t') + |1s'|+|1x'|}\\
			&=2\sum\limits_{s'\leq x'}(-1)^{s'\cdot t' + |s'|+|1x'|}
		\end{align*}
		
		Then, a quick induction gives us the final result $\sum\limits_{s\leq x}(-1)^{s\cdot t + |s|+|x|}= 2^{|x|}\delta_{x\leq t} $.
		
		So finally:
		
		\begin{center}
			$\tilde{f}(x)= f(\emptyset)^{\sum\limits_{s\leq x} (-1)^{|s|+|x|}}\prod\limits_{t}  \hat{f}(t)^{\sum\limits_{s\leq x} (s\cdot t)(-1)^{|s|+|x|}}
			= f(\emptyset)^{\delta_{x=\emptyset}}\prod\limits_{t}  \hat{f}(t)^{\frac{1}{2}\delta_{x=\emptyset}}\prod\limits_{t\geq x}  \hat{f}(t)^{-2^{|x|-1}}
			=\prod\limits_{t\geq x}  \hat{f}(t)^{-2^{|x|-1}}$
		\end{center}
		
		Since $\prod\limits_{t}  \hat{f}(t)^{\frac{1}{2}}=\prod\limits_{t}  \prod\limits_{x} f(x)^{-2^{-n}(-1)^{t\cdot x}}= \prod\limits_{x} f(x)^{-2^{-n}\sum\limits_{t} (-1)^{t\cdot x}}=\prod\limits_{x} f(x)^{-\delta_{0=x}}=f(\emptyset)^{-1}$.\\
		
		Now starting with the expression of the Fourrier transform:
		
		\begin{center}
			$\hat{f}(x)= \prod\limits_{t} f(t)^{-2^{1-n}(-1)^{x\cdot t}}
			= \prod\limits_{t}  \left(\prod\limits_{s\leq t}\tilde{f}(s)\right)^{-2^{1-n}(-1)^{x\cdot t}}
			= \prod\limits_{s}\tilde{f}(s)^{-2^{1-n}\sum\limits_{t\geq s}(-1)^{x\cdot t}} $
		\end{center}
		
		We evaluate $\sum\limits_{t\geq s}(-1)^{x\cdot t}$, if $x \geq t$ then:
		
		\begin{center}
			$\sum\limits_{t\geq s}(-1)^{x\cdot t}= (-1)^{|x\cap s|}\sum\limits_{t' }(-1)^{(x\cap \overline{s})\cdot t'}= (-1)^{|x\cap s|} 2^{n-|s|}\delta_{x\cap \overline{s}=\emptyset}= (-1)^{|x|} 2^{n-|s|}\delta_{x\leq s}$
		\end{center}
		
		Finally: $\hat{f}(x)= \prod\limits_{s}\tilde{f}(s)^{-(-1)^{|x|} 2^{1-|s|}\delta_{x\leq s}} = \prod\limits_{s\geq x}\tilde{f}(s)^{-2^{1-|s|}(-1)^{|x|}}$.
		
	\end{proof}
	
	\begin{proof}[Proof of Lemma \ref{lm:transformsSymBool}]$\quad$\\
		
		We start by computing the Fourrier transform:
		
		\begin{center}
			$\hat{f}(s)=\prod\limits_{t} f(t)^{-2^{1-n}(-1)^{s\cdot t}}=\prod\limits_{k=0}^n F(k)^{-2^{1-n}\sum\limits_{|s|=k}(-1)^{s\cdot t}}$.
		\end{center}
		
		By considering $j$, the number of ones in common between $x$ and $s$ we get:
		
		\begin{center}
			$\sum\limits_{|s|=k}(-1)^{s\cdot x}=\sum\limits_{j=0}^{k}\binom{|x|}{j}\binom{n-|x|}{k-j}(-1)^j = K^n_k(|x|)$.
		\end{center}
		
		For the inversion formula : $F(m)=F(0)\prod\limits_{k=0}^n \hat{F}(k)^{\frac{1}{2}\left(\binom{n}{k}-\sum\limits_{|s|=k}(-1)^{s\cdot t}\right)}=F(0)\prod\limits_{k=0}^n \hat{F}(k)^{\frac{1}{2}\left(\binom{n}{k}-K^n_k(m)\right)}$.
		
		Then using $F(0)^{-1}=\prod\limits_{k=0}^n \hat{F}(k)^{\frac{1}{2}\binom{n}{k}}$, we have: $F(m)=\prod\limits_{k=0}^n \hat{F}(k)^{\frac{-1}{2}K^n_k(m)}$.
		
		Now, for the Mobius transform, we compute: $\tilde{f}(x)=\prod\limits_{s\leq x} f(s)^{(-1)^{|s|+ |x|}}=\prod\limits_{k=0}^{|x|}F(k)^{\binom{|x|}{k}(-1)^{|x|-k}}$.
		
		And for the inversion formula, we get : $F(m)=\prod\limits_{s\leq x}\tilde{F}^{\binom{m}{k}}$.
		
	\end{proof}
	
	\begin{proof}[Proof of \ref{lm:nest}]
		
		We start with the identity :
		
		\begin{center}
			\input{mobiusgadget.tikz}
		\end{center}
		
		Here we show that it can be derived by a direct application of the Mobius diagrammatical transforms. To do so, we compute the Mobius transform of a phase gadget acting on $n$ qubits with parameter $\lambda$. The associated semi-boolean function, $X_{n,\lambda}: x \mapsto \lambda^{\frac{1-(-1)^{|x|}}{2}} $, is obviously symmetric.
		
		The binomial transform of $G: m \mapsto \lambda^{\frac{1-(-1)^{m}}{2}}$ is then:
		
		\begin{align*}
			\tilde{G}(m)&=\prod\limits_{k=0}^{m} G(k)^{\binom{m}{k} (-1)^{m-k}}=\prod\limits_{k=0}^{m} \lambda^{\frac{1-(-1)^{k}}{2} \binom{m}{k} (-1)^{m-k}}\\
			&=\lambda^{\frac{1}{2}\left(\sum\limits_{k=0}^{m} \binom{m}{k} (-1)^{m-k} - \sum\limits_{k=0}^{m} \binom{m}{k} (-1)^{m-k}(-1)^{k}\right)}=\lambda^{\frac{\delta_{m=0}-(-2)^{m}}{2}}=\lambda^{\delta_{m\neq0}(-2)^{m-1}}.
		\end{align*}
		
		Then, setting $\lambda= e^{\frac{i\pi}{4}}$ we get $\tilde{G}(0)=1$, so we have no floating scalars and also $\tilde{G}(m)=1$ for $m\geq 4$. The only non trivial terms are $\tilde{G}(1)=e^{\frac{i\pi}{4}}$, $\tilde{G}(2)=-i$ and $\tilde{G}(3)=-1$, directly giving the expected decomposition.\\
		
		We can compute the transform of the phase function $H(m)\df \beta\delta_{m=n}$ which corresponds to the generalized hyper-edges on $n$ qubits with phase $\beta$ :
		
		\begin{center}
			$\hat{H}(m)=\prod\limits_{k=0}^{n}H(k)^{-2^{1-n}K_k^n(m)}=\beta^{-2^{1-n}K_n^n(m)}=\beta^{-2^{1-n}(-1)^{m}}$.
		\end{center}
		
		Combining this result with the previous spider-nest identity, it is possible to derive the spider-nest identity from \cite{de2020fast} as it is done in \cite{munsonand}. However, we will here provide an alternative proof by inversion of the second identity.
		
		We first sketch the general method to check such spider-nest identity. We want to show that for a symmetric phase function $\hat{s}:\textbf{2}^n \to \mathbb{R}$:
		
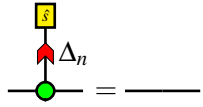
		\begin{center}
			$\input{invk0.tikz}=\input{invk1.tikz}$
		\end{center}
		
		We know $\hat{S}$ by reading the coefficients in the phase gadgets. We compute $S$ using the inversion formula. Then we check if $S$ is constant. If this is the case, this means that the corresponding phase gadget is $S(0)I_n$. But this scalar is exactly the one appearing in the graphical Fourier transform. So simplifying on both sides gives us exactly what we want.
		
		We apply this method to our specific case. Here, only the Kravchuk polynomials for $k=0,1,2,3$ and $n$ are needed:
		
		\begin{center}
			\begin{tabular}{cccc}
				$K^n_0(m)=1$&
				$K^n_1(m)=-2m+n$&
				$K^n_2(m)=2m^2 -2nm+\frac{n^2 -n}{2}$&$K^n_n(m)=(-1)^m$\\
				&&&\\
				\multicolumn{4}{c}{$K^n_3(m)=-\frac{4}{3}m^3 +2nm^2  +(-n^2 +n -\frac{2}{3})m+\frac{n^3 - 3n^2 + 2n}{6}$}
			\end{tabular}
		\end{center}
		
		We want to inverse the phase function:
		
		\begin{center}
			$\hat{S}(0)=1$ $\hat{S}(1)=e^{\frac{i\pi(n-2)(n-3)}{16}}$ $\hat{S}(2)=e^{\frac{-i\pi(n-3)}{8}}$ $\hat{S}(3)=e^{\frac{i\pi}{8}}$ $\hat{S}(n)=e^{\frac{-i\pi}{8}}$  $\hat{S}(k)=1$ for $ k\neq 0,1,2,3,n$.
		\end{center}
		
		and show that forall $m\in \interp{0,n}$, $S(m)= S(0)\mod 2$. The inversion formula gives us:
		
		\begin{center}
			\begin{align*}
				S(m)&=\prod\limits_{k=0}^{n}\hat{S}(k)^{\frac{-1}{2}K_k^n(m)}\\
				&=e^{i\pi\left(-\frac{(n-2)(n-3)}{32}K_1^n(m)\frac{n-3}{16}K_2^n(m)-\frac{1}{16}K_3^n(m)+\frac{1}{16}K_n^n(m)\right)}\\
				&=e^{i\pi\left(\frac{m^3}{12} -\frac{3m^2}{8} +\frac{5m}{12}- \frac{n^3}{96}+\frac{n^2}{16}-\frac{11n}{96} +\frac{(-1)^m}{16}\right)}
			\end{align*}
			
		\end{center}
		
		Our goal is to check that $S$ is constant. Thus, we only need the part of the exponent that depends on $m$.
		
		\begin{center}
			$E(m)\df\frac{m^3}{12} -\frac{3m^2}{8} +\frac{5m}{12}+\frac{(-1)^m}{16}$.
		\end{center} 
		
		Since $E(0)=\frac{1}{16} \mod 2$, we want to check that for each $m\in \mathbb{N}$, $E(m)\equiv \frac{1}{16} \mod 2$. To do so we write $m=24k+l$ with $k\in \mathbb{N}$ and $l\in \interp{0,23}$.
		We obtain:
		
		\begin{center}
			$E(24k+l)=1152k^3 + 144lk^2 - 216k^2 + 6kl^2 - 18kl +10k+ \frac{l^3}{12} -\frac{3l^2}{8}+\frac{5l}{12}+\frac{(-1)^{l}}{16}$
		\end{center}
		
		We see that $E(24k+l) = \frac{l^3}{12} -\frac{3l^2}{8}+\frac{5l}{12}+\frac{(-1)^{l}}{16} \mod 2$, this only depends on $l$. Thus we can just check that for each $l\in \interp{0,23}$, $\frac{l^3}{12} -\frac{3l^2}{8}+\frac{5l}{12}+\frac{(-1)^{l}}{16}= \frac{1}{16} \mod 2$  (which is true).
		
	\end{proof}
	
	\begin{proof}[Proof of Theorem \ref{thm:fhp}]
		
		We define $\Lambda(s)= \begin{cases}
			\lambda_i \text{ if } \exists i,\quad s=\overrightarrow{\{i\}}	\\
			1 \text{ else}
		\end{cases}$, we then have: $\mu_s = \prod\limits_{t\geq s}  \Lambda(t)^{-2^{|s|-1}}$.
		
		\begin{center}
			\input{hyperpivot.tikz}
		\end{center}
		
	\end{proof}

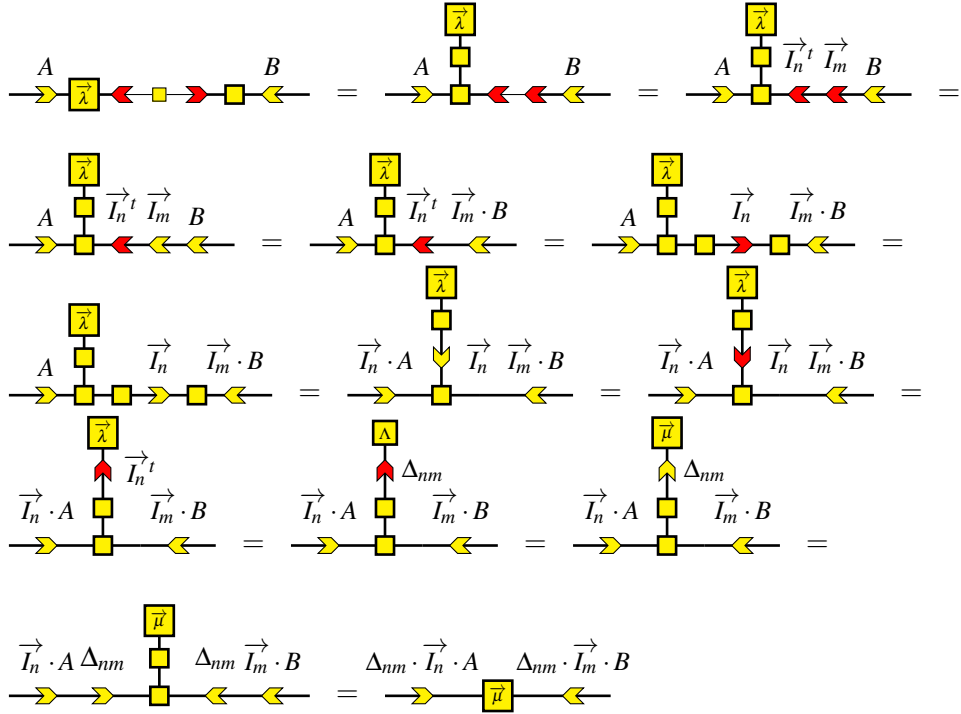
\end{document}

%% file: sample.tikzstyles
\tikzstyle{gn}=[fill=green, draw=black, shape=circle, tikzit category=ZX, tikzit fill=green, tikzit draw=black, tikzit shape=circle, inner sep=0.1em]
\tikzstyle{rn}=[fill=red, draw=black, shape=circle, tikzit fill=red, tikzit draw=black, tikzit category=ZX, tikzit shape=circle, inner sep=0.1em]
\tikzstyle{divide}=[regular polygon, regular polygon sides=3, shape border rotate=90, draw=black, fill={zx_grey}, inner sep=1.5pt, tikzit category=scal, rounded corners=0.8mm]
\tikzstyle{black}=[fill=black, draw=black, shape=circle, tikzit fill=black, tikzit draw=black, tikzit shape=circle, tikzit category=IH, inner sep=2pt]
\tikzstyle{gather}=[fill={zx_grey}, draw=black, tikzit category=scal, rounded corners=0.8mm, regular polygon, regular polygon sides=3, shape border rotate=-90, inner sep=1.5pt]
\tikzstyle{ggen}=[fill=white, draw=black, shape=rectangle, rounded corners=2mm, line width=1pt, tikzit draw=red, tikzit category=scal]
\tikzstyle{white}=[fill=white, draw=black, shape=circle, inner sep=2pt, tikzit category=IH]
\tikzstyle{mbox}=[fill=white, draw=black, rounded rectangle, rounded rectangle west arc=none, tikzit category=scal, tikzit shape=rectangle]
\tikzstyle{A}=[fill=white, shape=circle, tikzit category=scal, inner sep=1pt]
\tikzstyle{ggreen}=[fill=green, draw=black, shape=circle, tikzit category=SZX, tikzit fill=green, tikzit draw=black, line width=1pt, inner sep=0.1em]
\tikzstyle{gred}=[fill=red, draw=black, shape=circle, rounded corners=2mm, tikzit category=SZX, inner sep=0.1em, tikzit fill=red, line width=1pt]
\tikzstyle{ghad}=[minimum size=3mm, font={\scriptsize\boldmath}, shape=rectangle, inner sep=1mm, line width=1pt, outer sep=-1.5mm, scale=0.8, tikzit shape=rectangle, draw=black, fill=yellow, tikzit draw=blue]
\tikzstyle{boxm}=[fill=white, draw=black, rounded rectangle, tikzit category=scal, tikzit shape=rectangle, rounded rectangle east arc=none]
\tikzstyle{box}=[fill=white, draw=black, shape=rectangle]
\tikzstyle{had}=[fill=yellow, draw=black, shape=rectangle, tikzit category=ZX, tikzit fill=yellow, tikzit draw=black, inner sep=2.5pt, font={\scriptsize\boldmath}]
\tikzstyle{gwhite}=[fill=white, draw=black, shape=circle, tikzit fill=white, tikzit shape=circle, line width=1 pt, inner sep=2 pt, tikzit draw=red]
\tikzstyle{gblack}=[fill=black, draw=black, shape=circle, tikzit fill=black, tikzit shape=circle, line width=1 pt, inner sep=2 pt, tikzit draw=red]
\tikzstyle{antipode}=[fill=red, draw=black, shape=rectangle, tikzit fill=red, tikzit draw=black, tikzit shape=rectangle, inner sep=2pt]
\tikzstyle{diamond}=[fill=white, draw=black, shape=diamond, inner sep=2pt]
\tikzstyle{mongr}=[fill=green, draw=green, shape=circle, inner sep=2pt]
\tikzstyle{monbl}=[fill=blue, draw=blue, shape=circle, inner sep=2pt]
\tikzstyle{bg}=[inner sep=0.7mm, minimum width=0pt, minimum height=0pt, fill=green, draw=white, very thick, shape=circle]
\tikzstyle{br}=[inner sep=0.7mm, minimum width=0pt, minimum height=0pt, fill=red, draw=white, very thick, shape=circle]
\tikzstyle{rmat}=[draw, signal, fill=red, signal to=east, signal from=west, inner sep=1pt, minimum height=6pt]
\tikzstyle{lmat}=[draw, signal, fill=red, signal to=west, signal from=east, inner sep=1pt, minimum height=6pt]
\tikzstyle{umat}=[draw, signal, fill=red, signal to=north, signal from=south, inner sep=1pt, minimum width=6pt]
\tikzstyle{dmat}=[draw, signal, fill=red, signal to=south, signal from=north, inner sep=1pt, minimum width=6pt]
\tikzstyle{box}=[shape=rectangle, text height=1.5ex, text depth=0.25ex, yshift=0.5mm, fill=white, draw=black, minimum height=5mm, yshift=-0.5mm, minimum width=5mm, font={\small}]
\tikzstyle{Z dot}=[inner sep=0mm, minimum size=2mm, shape=circle, draw=black, fill={rgb,255: red,160; green,255; blue,160}]
\tikzstyle{gdot}=[minimum size=3mm, font={\scriptsize\boldmath}, shape=rectangle, rounded corners=1.3mm, inner sep=1mm, outer sep=-1.8mm, scale=0.8, tikzit shape=circle, draw=black, fill=green, tikzit draw=blue]
\tikzstyle{X dot}=[Z dot, shape=circle, draw=black, fill={rgb,255: red,220; green,0; blue,0}]
\tikzstyle{rdot}=[minimum size=3mm, font={\scriptsize\boldmath}, shape=rectangle, rounded corners=1.3mm, inner sep=1mm, outer sep=-1.8mm, scale=0.8, tikzit shape=circle, draw=black, fill=red, tikzit draw=blue]
\tikzstyle{grdot}=[minimum size=3mm, font={\scriptsize\boldmath}, shape=rectangle, rounded corners=1.3mm, inner sep=1mm, line width=1pt, outer sep=-1.5mm, scale=0.8, tikzit shape=circle, draw=black, fill=red, tikzit draw=blue]
\tikzstyle{ggdot}=[minimum size=3mm, font={\scriptsize\boldmath}, shape=rectangle, line width=1pt, rounded corners=1.3mm, inner sep=1mm, outer sep=-1.5mm, scale=0.8, tikzit shape=circle, draw=black, fill=green, tikzit draw=blue]
\tikzstyle{arrow}=[-->]
\tikzstyle{rfarr}=[draw, signal, fill=black, signal to=east, signal from=west, inner sep=1pt, minimum height=6pt]
\tikzstyle{lfarr}=[draw, signal, fill=black, signal to=west, signal from=east, inner sep=1pt, minimum height=6pt]
\tikzstyle{ufarr}=[draw, signal, fill=black, signal to=north, signal from=south, inner sep=1pt, minimum width=6pt]
\tikzstyle{dfarr}=[draw, signal, fill=black, signal to=south, signal from=north, inner sep=1pt, minimum width=6pt]
\tikzstyle{ry}=[draw, signal, fill=yellow, signal to=east, signal from=west, inner sep=1pt, minimum height=6pt]
\tikzstyle{ly}=[draw, signal, fill=yellow, signal to=west, signal from=east, inner sep=1pt, minimum height=6pt]
\tikzstyle{uy}=[draw, signal, fill=yellow, signal to=north, signal from=south, inner sep=1pt, minimum width=6pt]
\tikzstyle{dy}=[draw, signal, fill=yellow, signal to=south, signal from=north, inner sep=1pt, minimum width=6pt]

\tikzstyle{arrow}=[->]
\tikzstyle{very thick}=[-, line width=1pt, tikzit draw=red]
\tikzstyle{pointille}=[dashed, -]
\tikzstyle{red}=[-, draw=red]
\tikzstyle{blue}=[-, draw=blue]
\tikzstyle{green}=[-, draw=green]
\tikzstyle{strike}=[-, tikzit draw={rgb,255: red,191; green,0; blue,64}, strike through]
\tikzstyle{strike'}=[-, tikzit draw=cyan, strike bend]
\tikzstyle{dashed arrow}=[->, tikzit draw=green, draw=black, dashed]
\tikzstyle{reprise}=[-, line width=2pt, tikzit draw={rgb,255: red,255; green,191; blue,191}]
\tikzstyle{greyfill}=[-, fill={rgb,255: red,191; green,191; blue,191}]

%% file: figures/id.tikz
\begin{tikzpicture}
	\begin{pgfonlayer}{nodelayer}
		\node [style=none] (0) at (-1, 0) {};
		\node [style=none] (1) at (0.5, 0) {};
	\end{pgfonlayer}
	\begin{pgfonlayer}{edgelayer}
		\draw (0.center) to (1.center);
	\end{pgfonlayer}
\end{tikzpicture}

%% file: figures/cup.tikz
\begin{tikzpicture}
	\begin{pgfonlayer}{nodelayer}
		\node [style=none] (0) at (0, 0.75) {};
		\node [style=none] (1) at (0, -0.75) {};
	\end{pgfonlayer}
	\begin{pgfonlayer}{edgelayer}
		\draw [bend right=90, looseness=1.75] (0.center) to (1.center);
	\end{pgfonlayer}
\end{tikzpicture}

%% file: figures/cap.tikz
\begin{tikzpicture}
	\begin{pgfonlayer}{nodelayer}

		\node [style=none] (0) at (-0.5, -0.75) {};
		\node [style=none] (1) at (-0.5, 0.75) {};
	\end{pgfonlayer}
	\begin{pgfonlayer}{edgelayer}
		\draw [bend right=90, looseness=1.75] (0.center) to (1.center);
	\end{pgfonlayer}
\end{tikzpicture}

%% file: figures/def-bang-0.tikz
\begin{tikzpicture}
	\begin{pgfonlayer}{nodelayer}
		\node [style=X dot] (0) at (-1, -1.25) {};
		\node [style=had] (1) at (1, -1.25) {};
	\end{pgfonlayer}
\end{tikzpicture}

%% file: figures/def-bang-1.tikz
\begin{tikzpicture}
	\begin{pgfonlayer}{nodelayer}
		\node [style=X dot] (0) at (-1, -1.25) {};
		\node [style=had] (1) at (1, -1.25) {};
		\node [style=Z dot] (2) at (0, -0.25) {};
		\node [style=none] (3) at (0, 0.75) {};
	\end{pgfonlayer}
	\begin{pgfonlayer}{edgelayer}
		\draw (0) to (2);
		\draw (2) to (1);
		\draw (2) to (3.center);
	\end{pgfonlayer}
\end{tikzpicture}

%% file: figures/ex1.tikz
\begin{tikzpicture}
	\begin{pgfonlayer}{nodelayer}
		\node [style=none] (41) at (1, 0) {};
		\node [style=none] (42) at (0, 0) {};
	\end{pgfonlayer}
	\begin{pgfonlayer}{edgelayer}
		\draw [style=very thick] (42.center) to (41.center);
	\end{pgfonlayer}
\end{tikzpicture}

%% file: figures/el1.tikz
\begin{tikzpicture}
	\begin{pgfonlayer}{nodelayer}
		\node [style=none] (41) at (1, -0.5) {};
		\node [style=none] (42) at (0, -0.5) {};
		\node [style=none] (43) at (1, 0.5) {};
		\node [style=none] (44) at (0, 0.5) {};
	\end{pgfonlayer}
	\begin{pgfonlayer}{edgelayer}
		\draw [style=very thick] (42.center) to (41.center);
		\draw (44.center) to (43.center);
	\end{pgfonlayer}
\end{tikzpicture}

%% file: figures/n.tikz
\begin{tikzpicture}
	\begin{pgfonlayer}{nodelayer}
		\node [style=none] (31) at (1, 0) {};
		\node [style=none] (32) at (0, 0) {};
	\end{pgfonlayer}
	\begin{pgfonlayer}{edgelayer}
		\draw [style=very thick] (31.center) to (32.center);
	\end{pgfonlayer}
\end{tikzpicture}

%% file: figures/bcup.tikz
\begin{tikzpicture}
	\begin{pgfonlayer}{nodelayer}
		\node [style=none] (0) at (0, -0.5) {};
		\node [style=none] (1) at (0, 0.5) {};
	\end{pgfonlayer}
	\begin{pgfonlayer}{edgelayer}
		\draw [style=very thick, bend left=90, looseness=1.75] (0.center) to (1.center);
	\end{pgfonlayer}
\end{tikzpicture}

%% file: figures/bcap.tikz
\begin{tikzpicture}
	\begin{pgfonlayer}{nodelayer}
		\node [style=none] (0) at (0, -0.5) {};
		\node [style=none] (1) at (0, 0.5) {};
	\end{pgfonlayer}
	\begin{pgfonlayer}{edgelayer}
		\draw [style=very thick, bend right=90, looseness=1.75] (0.center) to (1.center);
	\end{pgfonlayer}
\end{tikzpicture}

%% file: figures/graphstate.tikz
\begin{tikzpicture}
	\begin{pgfonlayer}{nodelayer}
		\node [style=none] (228) at (1, 0.75) {$K_n$};
		\node [style=ly] (259) at (1, 0) {};
		\node [style=ghad] (260) at (0, 0) {$G$};
		\node [style=none] (261) at (2, 0) {};
	\end{pgfonlayer}
	\begin{pgfonlayer}{edgelayer}
		\draw [style=very thick] (261.center) to (260);
	\end{pgfonlayer}
\end{tikzpicture}

%% file: figures/invk1.tikz
\begin{tikzpicture}
	\begin{pgfonlayer}{nodelayer}
		\node [style=none] (86) at (2, 0) {};
		\node [style=none] (92) at (0, 0) {};
		\node [style=none] (96) at (1, 0) {};
	\end{pgfonlayer}
	\begin{pgfonlayer}{edgelayer}
		\draw [style=very thick] (92.center) to (96.center);
		\draw [style=very thick] (96.center) to (86.center);
	\end{pgfonlayer}
\end{tikzpicture}